\documentclass[copyright,creativecommons]{eptcs}
\providecommand{\event}{QPL 2013}
\usepackage{amsmath,amsthm,amssymb,xspace}
\usepackage[all]{xy}

\title{Categorical characterizations of operator-valued measures}
\author{Frank Roumen
\institute{Inst. for Mathematics, Astrophysics and Particle Physics (IMAPP) \\
Radboud University Nijmegen}
\email{F.Roumen@math.ru.nl}}

\def\titlerunning{Categorical characterizations of operator-valued measures}
\def\authorrunning{Frank Roumen}

\newcommand{\blank}{-}
\newcommand{\dd}{\,\mathrm{d}}
\newcommand{\der}[2]{\frac{\dd #1}{\dd #2}}
\newcommand{\Linf}{L^{\infty}}
\newcommand{\id}{\mathrm{id}}
\newcommand{\indicator}[1]{\ensuremath{\mathbf{1}_{#1}}}
\newcommand{\ket}[1]{\ensuremath{|{\kern.1em}#1{\kern.1em}\rangle}}
\newcommand{\bra}[1]{\ensuremath{\langle\,#1\,|}}

\newcommand{\Pred}{\HomKl(\blank,2)}
\newcommand{\C}{\mathbb{C}}
\renewcommand{\S}{\mathbb{S}}

\newcommand{\Cat}[1]{\ensuremath{\mathbf{#1}}}

\newcommand{\EA}{\Cat{EA}\xspace}
\newcommand{\EMod}{\Cat{EMod}\xspace}
\newcommand{\sEA}{\Cat{\sigma EA}\xspace}
\newcommand{\sEMod}{\Cat{\sigma EMod}\xspace}
\newcommand{\Meas}{\Cat{Meas}\xspace}
\newcommand{\Measure}{\Cat{Measure}\xspace}

\newcommand{\Hilbisomet}{\Cat{Hilb_{isomet}}\xspace}
\newcommand{\BNS}{\Cat{BNS}}
\newcommand{\Cstar}{\Cat{vN}\xspace}
\newcommand{\CstarPU}{\Cat{vN_{PU}}\xspace}

\newcommand{\Alg}{\Cat{Alg}}
\newcommand{\Kl}{\Cat{Kl}}
\newcommand{\comma}[2]{( #1 \downarrow #2)} 
\newcommand{\op}{^{\mathrm{op}}} 

\newcommand{\Giry}{\mathcal{G}}

\newcommand{\B}{\mathcal{B}}

\newcommand{\Pos}{\mathcal{P}\!os}
\newcommand{\Ef}{\mathcal{E}\!f}
\newcommand{\Proj}{\mathcal{P}\!roj}
\newcommand{\TC}{\mathcal{T}}
\newcommand{\DM}{\mathcal{D\!M}}

\DeclareMathOperator{\Hom}{Hom}
\DeclareMathOperator{\HomKl}{Hom_{\mathbf{Kl}}}

\DeclareMathOperator{\tr}{tr}

\newtheorem{theorem}{Theorem}
\newtheorem{lemma}[theorem]{Lemma}
\newtheorem{corollary}[theorem]{Corollary}
\newtheorem{proposition}[theorem]{Proposition}
\theoremstyle{definition}
\newtheorem{definition}[theorem]{Definition}
\newtheorem{example}[theorem]{Example}

\theoremstyle{remark}

\begin{document}

\maketitle

\begin{abstract}
    The most general type of measurement in quantum physics is modeled by a
    positive operator-valued measure (POVM). Mathematically, a POVM is a
    generalization of a measure, whose values are not real numbers, but
    positive operators on a Hilbert space. POVMs can equivalently be viewed as
    maps between effect algebras or as maps between algebras for the Giry
    monad. We will show that this equivalence is an instance of a duality
    between two categories. In the special case of continuous POVMs, we obtain
    two equivalent representations in terms of morphisms between von Neumann
    algebras.
\end{abstract}

\section{Introduction}\label{SecIntro}

The logic governing quantum measurements differs from classical logic, and it
is still unknown which mathematical structure is the best description of
quantum logic. The first attempt for such a logic was discussed in the famous
paper \cite{BirkhoffN36}, in which Birkhoff and von Neumann propose to use the
orthomodular lattice of projections on a Hilbert space. However, this approach
has been criticized for its lack of generality, see for instance
\cite{Schroeck96} for an overview of experiments that do not fit in the
Birkhoff-von Neumann scheme. The operational approach to quantum physics
generalizes the approach based on projective measurements. In this approach,
all measurements should be formulated in terms of the outcome statistics of
experiments. Thus the logical and probabilistic aspects of quantum mechanics
are combined into a unified description.

The basic concept of operational quantum mechanics is an effect on a Hilbert
space, which is a positive operator lying below the identity. It can be viewed
as a probabilistic version of a projection. The logical interpretation of an
effect is a predicate, or equivalently, a measurement with two possible
results. The logic of effects is useful in describing the semantics of quantum
programs via weakest preconditions, as argued in \cite{HondtP06}. A more
general treatment of the logical aspects of effects is given in
\cite{Jacobs12}. Both references use a duality between effects and convex sets
to relate syntax and semantics of the logic.

More generally, measurements with an arbitrary space of results can be modeled
as maps from the outcome space to the set of effects on a Hilbert space. These
maps are called positive operator-valued measures, or POVMs. This paper
presents several equivalent characterizations of POVMs, some of them
well-known, and some of them new. The results generalize the duality between
effects and convex sets. Thus they give a foundation for the connection
between syntax and semantics for a quantum logic where the predicates are
multivalued instead of two-valued.

The outline of this paper is as follows. Section~\ref{SecPreliminaries}
contains preliminaries about effect algebras, measure theory, and duality
between $\sigma$-effect modules and algebras for the Giry monad. This is
applied in Section~\ref{SecDuality} to obtain a categorical rephrasing of the
equivalence between POVMs and statistical maps. This result is already known
in the literature, but our systematic use of the abstract duality puts it in a
broader perspective. In Section~\ref{SecSeqComposition} we will generalize the
sequential composition operation on effects to POVMs. It will turn out that
this only works for a certain class of POVMs, namely those that are
differentiable with respect to an ambient measure. This gives a motivation to
study these differentiable POVMs in Section~\ref{SecContinuousPOVM}. To obtain
a duality result for differentiable POVMs, we will view them as morphisms
between von Neumann algebras.

\section{Preliminaries}\label{SecPreliminaries}

An effect algebra consists of a set $X$ equipped with a partial binary
operation $\oplus$, a unary operation $(\blank)^{\bot}$ called
orthocomplement, and a constant $0 \in X$, subject to the following
conditions:
\begin{itemize}
    \item The operation $\oplus$ is commutative, which means that whenever
        $x\oplus y$ is defined, also $y\oplus x$ is defined and $y\oplus x =
        x\oplus y$.
    \item The operation $\oplus$ is associative, defined in a similar way.
    \item $x \oplus 0 = 0 \oplus x = x$ for all $x\in X$.
    \item For every $x\in X$, $x^{\bot}$ is the unique element for which $x
        \oplus x^{\bot} = 1$, where $1$ is defined as $0^{\bot}$.
    \item If $x \oplus 1$ is defined, then $x=0$.
\end{itemize}
Effect algebras constitute a category $\EA$, in which the morphisms are
functions preserving $\oplus$, $(\blank)^{\bot}$, and $0$.
Effect algebras originated in the study of quantum logics in \cite{FoulisB94},
and can be used to describe both the probabilistic and the logical aspects of
quantum mechanics. An overview of the theory of effect algebras is given in
\cite{DvurecenskijP00}.

The principal example of an effect algebra is the unit interval $[0,1]$.
Addition serves as a partially defined binary operation, and the
orthocomplement is given by $x^{\bot} = 1-x$. Another important example comes
from quantum logic. An effect on a Hilbert space $H$ is an operator $A : H \to
H$ for which $0 \leq A \leq \id$. The set $\Ef(H)$ of all effects on $H$ forms
an effect algebra, in which the partial binary operation is again addition,
and orthocomplement is $A^{\bot} = \id - A$. Furthermore each Boolean algebra
$B$ can be viewed as an effect algebra, where $x\oplus y$ is defined if and
only if $x \wedge y = 0$, and in that case $x\oplus y = x \vee y$. The
orthocomplement is simply the complement in $B$.

Some effect algebras carry additional structure, which leads to several
commonly used subcategories of $\EA$. First we will consider the subcategory
$\EMod$ of effect modules. An effect module is an effect algebra $X$ endowed
with a scalar multiplication $\cdot : [0,1] \times X \to X$, such that
\begin{itemize}
    \item $r\cdot (s\cdot x) = (rs) \cdot x$.
    \item If $r + s \leq 1$, then $(r+s)\cdot x = r\cdot x \oplus s\cdot x$.
    \item If $x\oplus y$ is defined, then $r \cdot (x \oplus y) = r\cdot x
        \oplus r\cdot y$.
    \item $1\cdot x = x$.
\end{itemize}
Effect modules were introduced in \cite{GudderP98} under the name `convex
effect algebras', and generalized in \cite{JacobsM12c} to modules over
arbitrary effect algebras with a monoid structure, rather than just over the
interval $[0,1]$.
Morphisms of effect modules are morphisms of effect algebras that additionally
preserve the scalar multiplication. From our three examples of effect
algebras, only $[0,1]$ and $\Ef(H)$ are effect modules.

If $X$ is any effect algebra, then we can define a partial order on $X$ by
setting $x \leq y$ if and only if $x \oplus z = y$ for some $z \in X$. The
algebra $X$ is said to be an $\sigma$-effect algebra if each countable chain
in $X$ has a join in $X$. This gives rise to a subcategory $\sEA$ of $\EA$ in
which the morphisms also preserve joins of countable chains. A $\sigma$-effect
algebra that is at the same time an effect module is called a $\sigma$-effect
module, and they constitute a category $\sEMod$. The unit interval and
$\Ef(H)$ are always $\sigma$-effect modules. A Boolean algebra is a
$\sigma$-effect algebra if and only if it is a $\sigma$-algebra.

Given two effect algebras $X$ and $Y$, one can form their tensor product $X
\otimes Y$ characterizing the bimorphisms out of $X \times Y$. This tensor
product can be used to construct free effect modules: for any effect algebra
$X$, the tensor product $[0,1] \otimes X$ is the free effect module generated
by $X$. The situation is more subtle for $\sigma$-effect algebras, because the
tensor product of two $\sigma$-effect algebras need not always exist. This
problem is discussed in \cite{Gudder98}. 

Effect algebras also occur in measure theory. A measurable space consists of a
set $X$ together with a $\sigma$-algebra of subsets of $X$, denoted
$\Sigma_X$. Measurable spaces constitute a category $\Meas$, in which the maps
from $X$ to $Y$ are functions $f : X \to Y$ for which $f^{-1}(\Sigma_Y)
\subseteq \Sigma_X$. As each $\sigma$-algebra is an effect algebra with
countable joins, there is a functor $\Sigma_{(\blank)} : \Meas \to \sEA\op$.
The $\sigma$-effect algebra $\Sigma_X$ can be turned into a $\sigma$-effect
module by taking the tensor product $[0,1] \otimes \Sigma_X$. In
\cite{Gudder98} it is shown that this tensor product exists and is isomorphic
to the algebra $\Meas(X,[0,1])$ of measurable functions from $X$ to the unit
interval. In other words, $\Meas(X,[0,1])$ is the free $\sigma$-effect module
generated by the $\sigma$-algebra $\Sigma_X$.

Giry initiated the categorical approach to measure and integration theory in
\cite{Giry82} by defining the Giry monad $\Giry$ on the category $\Meas$ as
$\Giry(X) = \sEA(\Sigma_X,[0,1])$. Thus the elements of $\Giry(X)$ are
probability measures. A measurable map $p : X \to [0,1]$ can be integrated
along a probability measure $\varphi \in \Giry(X)$ to obtain $\int p \dd
\varphi \in [0,1]$, sometimes written as $\int p(x) \dd x$ if $\varphi$ is
understood.

In \cite{Jacobs13} it is shown that there is a dual adjunction between
Eilenberg-Moore algebras for the Giry monad and $\sigma$-effect modules:
\begin{equation}
    \label{EqAdjGiryEMod}
    \xymatrix@R+1pc{\Alg(\Giry) \ar@/^/[rr]^{\Hom(\blank, [0,1])} & \bot &
    \sEMod\op \ar@/^/[ll]^{\Hom(\blank, [0,1])}}
\end{equation}
This gives a foundation for probabilistic and quantum logic, since a
$\Giry$-algebra can be considered as the state space of a system, and the
corresponding $\sigma$-effect module gives the predicates on that system.

\section{Duality for POVMs}\label{SecDuality}

Effects on a Hilbert space can be seen as yes-no questions about the physical
system represented by the Hilbert space. It is also possible to consider more
general questions, which have answers lying in an arbitrary measurable space.
These can be mathematically modeled by positive operator-valued measures.

\begin{definition}
    Let $(X, \Sigma_X)$ be a measurable space.  A \emph{positive
    operator-valued measure (POVM)} on $X$ is a morphism $\Sigma_X \to \Ef(H)$
    of $\sigma$-effect algebras.
    A POVM is a \emph{projection-valued measure (PVM)} if its image is
    contained in $\Proj(H)$.
\end{definition}

We will study these POVMs from the viewpoint of categorical logic. The syntax
of a logic is obtained by defining operations on predicates, leading to an
algebraic structure. For instance, the predicates in operational quantum logic
are effects on a Hilbert space, and the appropriate operations are the
$\sigma$-effect algebra operations. The semantics of a logic is related to the
syntax via duality. In our quantum example, the semantics is given by density
matrices, since density matrices and effects are related via the duality
between convex sets and effect algebras, see \cite{JacobsM12b} for details.

In the remainder of this paper, we will try to establish a similar picture for
POVMs. This section considers a generalization of the duality for effects to
POVMs. The duality for POVMs will be based on the adjunction
\eqref{EqAdjGiryEMod}, so it is helpful to rephrase the definition of POVMs in
terms of morphisms between modules. 

\begin{lemma}
    There is a bijective correspondence between POVMs $\Sigma_X \to \Ef(H)$
    and morphisms of $\sigma$-effect modules $\Meas(X,[0,1]) \to \Ef(H)$.
    \label{LemPOVMEMod}
\end{lemma}

This follows immediately from Gudder's result that $\Meas(X,[0,1])$ is the
free $\sigma$-effect module on $\Sigma_X$, which was briefly mentioned in
Section~\ref{SecPreliminaries}.  More explicitly, if $\varphi : \Sigma_X \to
\Ef(H)$ is a POVM, then the corresponding map $\Meas(X,[0,1]) \to \Ef(H)$ is
given by integration along the POVM $\varphi$, i.e. $p \mapsto \int p \dd
\varphi$.  The inverse construction is evaluation at an indicator function,
that is, a map $\Phi : \Meas(X,[0,1]) \to \Ef(H)$ gives a POVM $M \mapsto
\Phi(\indicator{M})$.

\begin{lemma}
    \label{LemDMfunctor}
    There is a functor $\DM : \Hilbisomet \to \Alg(\Giry)$ that maps a Hilbert
    space $H$ to the set of density matrices on $H$. Here $\Hilbisomet$ is the
    category with Hilbert spaces as objects and isometries as morphisms.
\end{lemma}
\begin{proof}
    First we have to endow $\DM(H)$ with the structure of a measurable space.
    The weak operator topology on $\B(H)$ restricts to a subset topology on
    $\DM(H)$. Let $\Sigma_{\DM(H)}$ be the Borel $\sigma$-algebra generated by
    the topology on $\DM(H)$. The resulting measurable space $(\DM(H),
    \Sigma_{\DM(H)})$ is an algebra for the Giry monad with algebra map
    $\alpha : \Giry(\DM(H)) \to \DM(H)$, $\alpha(\varphi) = \int \id \dd
    \varphi$.  Here the integration is defined in such a way that $\bra{\psi}
    \int \id \dd \varphi \ket{\psi} = \int \bra{\psi} (\blank) \ket{\psi} \dd
    \varphi$ for each vector $\psi$.
    The map $\alpha$ is measurable by general facts about integration.  To
    show that the integral is a density matrix, let $(e_k)$ be an orthonormal
    basis for $H$. Then
    \[ \begin{array}{c} \tr \left( \int \id \dd \varphi \right)
        =  \sum_k \bra{e_k} \int \id \dd \varphi \ket{e_k} 
        =  \sum_k \int \bra{e_k} (\blank) \ket{e_k} \dd \varphi 
        =  \int \sum_k \bra{e_k} (\blank) \ket{e_k} \dd \varphi 
        =  \int 1 \dd \varphi = 1.
    \end{array} \]
    Proving that $\alpha$ is an Eilenberg-Moore algebra is straightforward.

    If $f : H \to K$ is an isometry between Hilbert spaces, then $f$ induces a
    map $\DM(f) : \DM(H) \to \DM(K)$ via conjugation, i.e. $\DM(f)(\rho) = f
    \circ \rho \circ f^{\dag}$. The resulting map is a $\Giry$-algebra
    homomorphism. Before proving this, we first remark that conjugation
    commutes with integration in the sense that $\int (f (\blank) f^\dag) \dd
    \varphi = f \left( \int (\blank) \dd \varphi \right) f^\dag$. This follows
    because for each vector $\psi$ we have
    \[
    \begin{array}{rcl}
        \bra{\psi} \int f (\blank) f^\dag \dd \varphi \ket{\psi}
        &=& \int \bra{\psi} f (\blank) f^\dag \ket{\psi} \dd \varphi \\
        &=& \int \bra{f^\dag \psi} (\blank) \ket{f^\dag \psi} \dd \varphi \\
        &=& \bra{f^\dag \psi} \int (\blank) \dd \varphi \ket{f^\dag \psi} \\
        &=& \bra{\psi} f ( \int (\blank) \dd \varphi ) f^\dag \ket{\psi}
    \end{array}
    \]
    Using this we can show that $\DM(f)$ is a $\Giry$-algebra homomorphism:
    \[
    \begin{array}{rcl}
        (\alpha \circ \Giry(\DM(f)))(\varphi)
        &=& \int \id \dd \Giry(\DM(f))(\varphi) \\
        &=& \int \DM(f) \dd \varphi \\
        &=& \int f (\blank) f^\dag \dd \varphi \\
        &=& f \left( \int (\blank) \dd \varphi \right) f^{\dag} \\
        &=& \DM(f) \left( \int \id \dd \varphi \right) \\
        &=& (\DM(f) \circ \alpha) (\varphi) 
    \end{array} \]
    This shows that $\DM$ is a well-defined functor.
\end{proof}

The collections of density matrices and effects on a Hilbert space are related
via the adjunction \eqref{EqAdjGiryEMod}, just like in the discrete
probabilistic case.

\begin{proposition}
    \label{PropEffectStateDuality}
    Fix a Hilbert space $H$. Then:
    \begin{enumerate}
        \item The $\Giry$-algebras $\sEMod(\Ef(H), [0,1])$ and $\DM(H)$ are
            isomorphic.
        \item The $\sigma$-effect modules $\Alg(\Giry)(\DM(H), [0,1])$ and
            $\Ef(H)$ are isomorphic.
    \end{enumerate}
\end{proposition}
\begin{proof}\mbox{}
    \begin{enumerate}
        \item This is a reformulation of Busch's theorem in \cite{Busch03}.
        \item In \cite{JacobsM12c} this result is proven for affine maps
            $\DM(H) \to [0,1]$ instead of $\Giry$-algebra maps, so the
            statement follows because every $\Giry$-algebra map is in
            particular affine.  \qedhere

   \end{enumerate}
\end{proof}

Since $[0,1] \cong \Giry(2)$, measurable maps into $[0,1]$ are the same as
morphisms into $2$ in the Kleisli category $\Kl(\Giry)$.
The following diagram summarizes the relations between the logic of measurable
spaces and the logic of Hilbert spaces.

\[
\xymatrix@R+1pc { & \Kl(\Giry) \ar[dl]_{\mathcal{K}}
\ar[dr]^{\HomKl(\blank,2)} \\
\Alg(\Giry) \ar@/^/[rr]^{\Hom(\blank, [0,1])} & \bot &
\sEMod\op \ar@/^/[ll]^{\Hom(\blank, [0,1])} \\
& \Hilbisomet \ar[ul]^{\DM} \ar[ur]_{\Ef} }
\]
The functor $\mathcal{K}$ is the comparison functor sending an object $X \in
\Kl(\Giry)$ to the free algebra $\Giry(X)$.
In this setting we can consider the comma categories
$\comma{\DM}{\mathcal{K}}$ and $\comma{\Ef}{\HomKl(\blank,2)}$. An object of the
category $\comma{\DM}{\mathcal{K}}$ is a
map of the form $\DM(H) \to \Giry(X)$. A morphism from $ \alpha :
\DM(H) \to \Giry(X)$ to $\beta : \DM(K) \to \Giry(Y)$ is a commutative diagram
\begin{equation}
    \label{EqMorphismDMGiry}
    \begin{gathered}
        \xymatrix { \DM(H) \ar[r]^{\DM(g)} \ar[d]_\alpha &
        \DM(K) \ar[d]^\beta \\
        \Giry(X) \ar[r]_{\mathcal{K}(f)} & \Giry(Y) }
    \end{gathered}
\end{equation}
where $f: X \to \Giry(Y)$ is a measurable map and $g : H \to K$ is an
isometry. Since the functors $\Ef$ and $\Pred$ have the opposite of $\sEMod$
as codomain, an object of $\comma{\Ef}{\Pred}$ is a morphism $\HomKl(X,2) \to
\Ef(H)$ in $\sEMod$,
that is, a POVM. A morphism between two POVMs $ A : \HomKl(X,2) \to \Ef(H)$ and
$ B : \HomKl(Y,2) \to \Ef(K)$ is given by a diagram
\begin{equation}
    \label{EqMorphismEfSigma}
    \begin{gathered}
        \xymatrix { \HomKl(X,2) \ar[d]_A &
        \HomKl(Y,2) \ar[l]_{(\blank)\circ f} \ar[d]^B \\
        \Ef(H) & \Ef(K) \ar[l]^{\Ef(g)} }
    \end{gathered}
\end{equation}
in $\sEMod$, for a measurable map $f : X \to \Giry(Y)$ and an isometry $g : H
\to K$.

In \cite{HeinosaariZ12} it is shown that there is a correspondence between
POVMs and $\Giry$-algebra homomorphisms $\DM(H) \to \Giry(X)$, called
statistical maps. From a categorical perspective, this can be phrased as an
equivalence between comma categories as follows.

\begin{proposition}
    The categories $\comma{\DM}{\Giry}$ and $\comma{\Ef}{\Pred}$ are
    equivalent.
\end{proposition}
\begin{proof}
    An object of $\comma{\DM}{\Giry}$ is the same as a morphism $\DM(H)
    \to \sEMod(\HomKl(X,2),[0,1])$ since $\HomKl(X,2) \cong
    \Meas(X,[0,1])$ is the free $\sigma$-effect module on $\Sigma_X$. By the
    adjunction \eqref{EqAdjGiryEMod} and
    Proposition~\ref{PropEffectStateDuality}, this corresponds to a POVM.
    For morphisms, let $f : X
    \to \Giry(Y)$ be a measurable map, $g: H \to K$ an isometry, and $\alpha :
    \DM(H) \to \Giry(X)$ and $\beta : \DM(K) \to \Giry(Y)$ two statistical
    maps.  Then the diagram \eqref{EqMorphismDMGiry} commutes if and only if
    the corresponding diagram \eqref{EqMorphismEfSigma} commutes.
\end{proof}

\section{Sequential composition}\label{SecSeqComposition}

Suppose that we want to test two properties of a physical system sequentially.
If the properties are modeled by effects $A$ and $B$, then the composite test
corresponds to the effect $\sqrt{A} B \sqrt{A}$, which is called the
sequential product of $A$ and $B$. The properties of this operation are
studied in \cite{GudderN01,GudderG02,Gudder10}. We will now define an
extension of this operation to POVMs, which can be used if we want to measure
two POVMs sequentially. We start by measuring a POVM $A: \Sigma_X \to \Ef(H)$.
The outcome of this measurement is a value $x\in X$. The second POVM may
depend on the outcome of the first measurement, so we assume that we have a
family of measurable spaces $(Y_x)$ indexed by $x\in X$ with a family of POVMs
$B = (B_x : \Sigma_{Y_x} \to \Ef(H))$. We wish to define a POVM representing
the total experiment. For this we need the additional assumptions that the
measurable space $X$ is equipped with a finite measure $\mu : \Sigma_X \to
\mathbb{R}$, and that $A$ has a Radon-Nikodym derivative with respect to
$\mu$. Recall that a Radon-Nikodym derivative of $A$ with respect to $\mu$ is
a function $\der{A}{\mu} : X \to \Pos(H)$ for which $\int_M \der{A}{\mu} \dd
\mu = A(M)$ for each measurable subset $M\subseteq X$. Here $\Pos(H)$ denotes
the set of positive operators on $H$. The derivative, if it exists, is unique
up to equality almost everywhere. Conditions for existence are discussed in
e.g. \cite{DiestelU77}. Here we will only briefly state the result that we
need for the remainder of this paper. The POVM $A$ is called $\mu$-continuous
if $A(M) = 0$ whenever $\mu(M) = 0$. It has bounded variation if
\[ \sup \sum_{i=1}^n || \varphi(X_i) || < \infty, \]
where the supremum is taken over all finite partitions $X = \bigcup_{i=1}^n
X_i$ of the space $X$. If the Hilbert space $H$ is finite-dimensional, then
the POVM $A$ has a derivative if and only if it is $\mu$-continuous and has
bounded variation, because $\B(H)$ has the Radon-Nikodym property.

Under the assumption that $A$ has a Radon-Nikodym derivative $\der{A}{\mu}$,
we can define the sequential composition of the POVM $A$ and the
family $B$. The total outcome of the experiment consists of a value $x\in X$
together with a value $y \in Y_x$, so our outcome space is $Y = \bigcup_{x\in
X} Y_x$. The union carries a natural $\sigma$-algebra generated by
$\bigcup_{x\in M} N_x$, where $M$ is a measurable subset of $X$ and each $N_x$
is a measurable subset of $Y_x$. Define the sequential composition by
\[ ( A;B ) : \Sigma_Y \to \Ef(H) \]
\[ ( A;B ) \left( \bigcup_{x\in M} N_x \right)
= \int_M \left( \sqrt{\frac{\dd A}{\dd \mu}(x)}
B_x(N_x) \sqrt{\frac{\dd A}{\dd \mu}(x)} \right) \dd x \] 

\begin{lemma}
    The sequential composition $(A;B)$ is a POVM.
    \label{LemSeqComposition}
\end{lemma}
\begin{proof}
    Suppose that the measurable sets $\bigcup_{x\in M} N_x$ and $\bigcup_{x\in
    M'} N'_x$ are disjoint. Then their union can be written as
    \[ \bigcup_{x\in M \backslash M'} N_x \cup
    \bigcup_{x\in M' \backslash M} N'_x \cup
    \bigcup_{x\in M \cap M'} (N_x \cup N'_x), \]
    where $N_x$ and $N'_x$ are disjoint whenever both are defined. Applying the
    map $A;B$ gives
    \[ \begin{array}{rcl}
        (A;B)(\bigcup_{x\in M} N_x \cup \bigcup_{x\in M'} N'_x)
        &=& \int_{M \backslash M'} \sqrt{\frac{\dd A}{\dd
        \mu}(x)} B_x(N_x) \sqrt{\frac{\dd A}{\dd \mu}(x)} \dd x \\
        & & + \int_{M' \backslash M} \sqrt{\frac{\dd A}{\dd
        \mu}(x)} B_x(N'_x) \sqrt{\frac{\dd A}{\dd \mu}(x)} \dd x \\
        & & + \int_{M \cap M'} \sqrt{\frac{\dd A}{\dd \mu}(x)} B_x(N_x \cup
        N'_x) \sqrt{\frac{\dd A}{\dd \mu}(x)} \dd x \\
        &=& \int_{M \backslash M'} \sqrt{\frac{\dd A}{\dd
        \mu}(x)} B_x(N_x) \sqrt{\frac{\dd A}{\dd \mu}(x)} \dd x \\
        & & + \int_{M' \backslash M} \sqrt{\frac{\dd A}{\dd
        \mu}(x)} B_x(N'_x) \sqrt{\frac{\dd A}{\dd \mu}(x)} \dd x \\
        & & + \int_{M \cap M'} \sqrt{\frac{\dd A}{\dd \mu}(x)} B_x(N_x)
        \sqrt{\frac{\dd A}{\dd \mu}(x)} \dd x \\
        & & + \int_{M \cap M'} \sqrt{\frac{\dd A}{\dd \mu}(x)} B_x(N'_x)
        \sqrt{\frac{\dd A}{\dd \mu}(x)} \dd x \\
        &=& \int_M \sqrt{\frac{\dd A}{\dd \mu}(x)} B_x(N_x) \sqrt{\frac{\dd
        A}{\dd \mu}(x)} \dd x \\
        & & + \int_{M'} \sqrt{\frac{\dd A}{\dd \mu}(x)} B_x(N'_x)
        \sqrt{\frac{\dd A}{\dd \mu}(x)} \dd x \\
        &=& (A;B)(\bigcup_{x\in M} N_x) + (A;B)(\bigcup_{x\in M'} N'_x)
    \end{array} \]
    Hence the map $A;B$ is additive. It is not hard to check that it preserves
    the unit.
    Finally, each operator $(A;B)(\bigcup_{x \in M} N_x)$ is positive, and
    lies below the identity because $(A;B)(\bigcup_{x \in M} N_x) \leq
    (A;B)(Y) = \id$. Thus $A;B$ is a POVM. 
\end{proof}

\begin{example}
    We apply the above construction to the spin example from
    \cite{HeinosaariZ12}.
    Consider a system consisting of one spin-$\frac{1}{2}$ particle, modeled
    as the Hilbert space $\C^2$. The direction of the spin has a value in the
    unit sphere $\S^2$, and is given by the POVM 
    \[ D : \Sigma_{\S^2} \to \Ef(\C^2) \]
    \[ D(M) = \frac{1}{4\pi} \int_M (\id + \vec{n} \cdot
    \vec{\sigma}) \dd \vec{n} \]
    Here $\dd \vec{n}$ is the usual measure on the unit sphere, and
    $\vec{\sigma} = (\sigma_x, \sigma_y, \sigma_z)$ is the vector consisting
    of the Pauli matrices:
    \[ \sigma_x = \left(
    \begin{array}{cc}
        0 & 1 \\ 1 & 0
    \end{array}
    \right), \qquad
    \sigma_y = \left(
    \begin{array}{cc}
        0 & -i \\
        i & 0 
    \end{array}
    \right), \qquad
    \sigma_z = \left(
    \begin{array}{cc}
        1 & 0 \\ 0 & -1
    \end{array}
    \right). \]

    If we pick a direction $\vec{n} \in \S^2$, then we can also measure the
    spin component along the direction $\vec{n}$. This measurement has two
    possible outcomes, which we label by $+$ and $-$. The corresponding POVM
    is $S_{\vec{n}} : \Sigma_{\{\pm\}} \to \Ef(\C^2)$, defined by
    $S_{\vec{n}}(\{\pm\}) = \frac{1}{2}(\id \pm \vec{n} \cdot
    \vec{\sigma})$. Physically, the probability that the outcome is $+$
    indicates how close the actual spin direction of the particle is to
    $\vec{n}$.

    We perform the following experiment on the system. First we measure the
    spin direction, which has outcome $\vec{n}$. Then we measure the spin
    component along this direction, i.e. we perform the measurement
    $S_{\vec{n}}$. Since the spin direction of the particle is in this
    situation equal to the measurement direction, we expect that the second
    measurement always gives outcome $+$. The outcome space of the composite
    measurement $D;S$ is $\bigcup_{\vec{n} \in \S^2} \{ \pm \} \cong \S^2
    \times \{ \pm \}$.
    According to the physical interpretation, this composite measurement is
    determined by
    \[ (D ; S)(M\times \{-\}) = 0 \]
    \[ (D ; S)(M\times \{+\}) = D(M) \]

    We can also verify this using the sequential composition formula.
    From the definition of the POVM $D$ it is immediate that its Radon-Nikodym
    derivative is
    \[ \frac{\dd D}{\dd \vec{n}}\left( \vec{n} \right) = \frac{1}{4\pi}
    (\id + \vec{n} \cdot \vec{\sigma}). \]
    Then the `minus' case of the sequential composition formula becomes:
    \[ \begin{array}{rcl}
        (D;S)(M \times \{-\}) &=& \int_M \sqrt{\frac{\dd D}{\dd \vec{n}}
        (\vec{n})} S_{\vec{n}}(\{-\}) \sqrt{\frac{\dd D}{\dd \vec{n}}(\vec{n})}
        \dd \vec{n} \\
        &=& \frac{1}{8\pi} \int_M \left( \sqrt{\id + \vec{n} \cdot
        \vec{\sigma}} \left( \id - \vec{n} \cdot \vec{\sigma} \right)
        \sqrt{\id + \vec{n} \cdot \vec{\sigma}} \right) \dd \vec{n} \\
        &=& \frac{1}{8\pi} \int_M (\id + \vec{n} \cdot \vec{\sigma})
        (\id - \vec{n} \cdot \vec{\sigma}) \dd \vec{n} \\
        &=& \frac{1}{8\pi} \int_M (\id - (\vec{n} \cdot \vec{\sigma})^2)
        \dd \vec{n}
    \end{array} \]
    For the third equality sign, we used that a square root $\sqrt{A}$
    commutes with every operator that commutes with $A$. A well-known
    property of the Pauli matrices is that $(\vec{n} \cdot \vec{\sigma})^2 =
    \id$ for each unit vector $\vec{n}$. From this it follows that $(D;S)(M
    \times \{-\}) = 0$. An analogous computation shows that $(D;S)(M \times
    \{+\}) = D(M)$.
    \label{ExSpinMeasurement}
\end{example}

\section{Characterization of continuous POVMs}\label{SecContinuousPOVM}

In Section~\ref{SecSeqComposition} we saw that we need continuity conditions
on POVMs in order to define sequential composition. Therefore we will now
study continuous POVMs in more detail and provide a few equivalent
characterizations. It will turn out that in the continuous case von Neumann
algebras form a more natural setting than effect algebras. Our main examples
of von Neumann algebras are constructed from Hilbert spaces and measure
spaces. If $H$ is a Hilbert space, then $\B(H)$ will denote the von Neumann
algebra of bounded linear operators on $H$. Recall that a measure space is a
measurable space together with a measure. For a measure space $(X,\mu)$, let
$\Linf(X,\mu)$ be the algebra of $\mu$-essentially bounded functions from $X$
to $\C$, modulo equality almost everywhere. We will assume throughout this
section that $X$ arises from a compact Hausdorff space and that $\mu(X)$ is
finite.

The duality for non-continuous POVMs boiled down to the duality between states
and effects. For continuous POVMs we will replace this by the interplay
between a von Neumann algebra and its normal states, or its predual. To
describe this in more detail, we will use several categories of von Neumann
algebras. The standard notion of morphism between C*-algebras is a
$*$-homomorphism, which is a bounded linear map preserving multiplication,
unit, and involution.  For von Neumann algebras we usually impose an
additional condition: a map between von Neumann algebras is called normal if
it preserves joins of countable increasing chains. This is equivalent to
preservation of countable sums of orthogonal projections, see e.g.
\cite{KadisonR83} for details. The category of unital von Neumann algebras
with normal $*$-homomorphisms will be denoted $\Cstar$. Sometimes it is more
appropriate to use a weaker notion of morphism. The category with von Neumann
algebras as objects and normal linear maps preserving positivity and the unit
as morphisms is denoted $\CstarPU$.

The predual $A_{\#}$ of a von Neumann algebra $A$ consists of all normal
linear functionals from $A$ to $\C$. If $A$ is unital, then the predual is
equipped with a canonical trace map $\tau : A_{\#} \to \C$, given by
evaluation at the unit. For example, the predual of $\B(H)$ is the collection
of trace-class operators $\TC(H)$, and the canonical trace map is the ordinary
trace $\tr : \TC(H) \to \C$. The predual of $\Linf(X,\mu)$ is $L^1(X,\mu)$,
i.e. the measurable functions $f: X \to \C$ such that the integral $\int_X |f|
\dd \mu$ is finite. In this case, the trace map is integration $\int_X
(\blank) \dd \mu$.

The structure of a predual can be captured abstractly by base norm spaces, see
e.g. \cite{Alfsen71,Nagel74}. Let $V$ be an ordered vector space, and $\tau :
V \to \C$ a positive linear functional. A convex subset $C$
of $V$ is called linearly bounded if $C \cap L$ is bounded for every line $L$
through the origin. Let $K = \tau^{-1}(1) \subseteq V$; the pair $(V,\tau)$ is
said to be a base norm space if the convex hull of $K \cup -K$ is linearly
bounded. Base norm spaces form a category $\BNS$ in which a morphism from
$(V,\tau)$ to $(V',\tau')$ is a positive linear map $f : V \to V'$ for which
$\tau' \circ f = \tau$.

The following result shows how to view continuous POVMs as morphisms between
von Neumann algebras.

\begin{proposition}
    There is a bijective correspondence between:
    \begin{itemize}
        \item POVMs $\Sigma_X \to \Ef(H)$ that are $\mu$-continuous and have
            bounded variation;
        \item Normal positive unital maps $\Linf(X,\mu) \to \B(H)$.
    \end{itemize}
    \label{PropBijectionPOVMCStar}
\end{proposition}
\begin{proof}
    Let $\varphi : \Sigma_X \to \Ef(H)$ be a POVM.
    Define a map
    $\psi_\varphi : \Linf(X, \mu) \to \B(H)$ by $\psi_\varphi(f) = \int_X f
    \dd \varphi$. This integral is well-defined since $f$ is essentially
    bounded. To verify that the map $\psi_\varphi$ is well-defined, we have to
    check that it maps functions that are zero almost everywhere to the zero
    operator. If an indicator function $\indicator{M}$ is zero almost
    everywhere, then $\mu(M) = 0$, so from $\mu$-continuity of $\varphi$ it
    follows that $\int \indicator{M} \dd \varphi = \varphi(M) = 0$. For
    general functions in $\Linf(X, \mu)$ this follows from linearity and
    continuity of the integral.
    Furthermore the map $\psi_\varphi$ is positive and unital. It preserves
    joins of countable chains since $\varphi$ is a POVM. Every positive map
    between von Neumann algebras is bounded, see \cite[Prop.
    1.3.7]{Landsman98} for a proof.

    In the other direction, given a map $\psi : \Linf(X,\mu) \to \B(H)$,
    define $\varphi_\psi : \Sigma_X \to \Ef(H)$ by $\varphi_\psi(M) =
    \psi(\indicator{M})$. Then $\varphi_\psi(M)$ is positive because $\psi$
    preserves positivity, and $\varphi_\psi(M) \leq \psi(\indicator{X}) =
    \id$, so $\varphi_\psi(M)$ is an effect. The map $\varphi_\psi$ is a
    morphism of $\sigma$-effect algebras since $\psi$ is linear, normal, and
    unital. To establish $\mu$-continuity of $\varphi_\psi$, suppose that
    $\mu(M) = 0$. Then $\indicator{M}$ is zero almost everywhere, hence
    $\varphi_\psi(M) = \psi(\indicator{M}) = \psi(0) = 0$.
    Finally, $\varphi_\psi$ has bounded variation because
    \[ \begin{array}{c}
        \sup \sum_i || \varphi_\psi(X_i) ||
        = \sup \sum_i || \psi(\indicator{X_i}) ||
        \leq \sup \sum_i ||\psi || \mu(X_i)
        = || \psi || \mu(X) < \infty.
    \end{array} \]
    It is easy to see that both constructions are inverses.
\end{proof}

Observe that the construction of the map between von Neumann algebras from a
POVM did not use the fact that the POVM has bounded variation. Thus we obtain
the following consequence.

\begin{corollary}
    Every $\mu$-continuous POVM has bounded variation.
    \label{CorPOVMBoundedVar}
\end{corollary}

Therefore we can simply work with $\mu$-continuous POVMs from now on, ignoring
the condition on the variation.  It is also possible to characterize
projection-valued measures as maps between von Neumann algebras, by
restricting the above correspondence.

\begin{corollary}
    There is a bijective correspondence between:
    \begin{itemize}
        \item PVMs $\Sigma_X \to \Proj(H)$ that are $\mu$-continuous;
        \item Normal $*$-homomorphisms $\Linf(X,\mu) \to \B(H)$.
    \end{itemize}
    \label{CorBijectionPVMCStar}
\end{corollary}
\begin{proof}
    Let $\psi : \Linf(X,\mu) \to \B(H)$ be a normal $*$-homomorphism. By
    Proposition~\ref{PropBijectionPOVMCStar}, it gives a $\mu$-continuous
    POVM $\varphi : \Sigma_X \to \Ef(H)$. We have to check that each
    $\varphi(M)$ is a projection:
    \[ \varphi(M)^2 = \psi(\indicator{M})^2 = \psi( (\indicator{M})^2) =
    \psi(\indicator{M}) = \varphi(M), \]
    where we used that $\psi$ preserves multiplication in the second equality
    sign.

    Conversely, a $\mu$-continuous PVM $\varphi$ gives a normal positive
    unital map $\psi : \Linf(X,\mu) \to \B(H)$. To show that $\psi$ preserves
    multiplication, we start by considering indicator functions:
    \[ \psi(\indicator{M} \indicator{N}) = \psi(\indicator{M \cap N}) =
    \varphi(M\cap N) = \varphi(M) \varphi(N) = \psi(\indicator{M})
    \psi(\indicator{N}) \]
    The third equality sign is a property that characterizes the
    projection-valued measures. All essentially bounded functions from $X$ to
    $\C$ can be written as a countable join of sums of indicator functions,
    modulo equality almost everywhere. Since $\psi$ preserves sums and
    countable joins, it follows that $\psi(fg) = \psi(f) \psi(g)$ for all $f$
    and $g$.
\end{proof}

The characterization of continuous POVMs as maps between the von Neumann
algebras $\Linf(X,\mu)$ and $\B(H)$ is in line with the Heisenberg picture of
quantum mechanics. There is also a characterization from the Schr\"odinger
point of view, analogous to considering POVMs as maps between
$\Giry$-algebras.

\begin{proposition}
    \label{PropDualityCStar}
    There is a bijective correspondence between:
    \begin{itemize}
        \item Normal positive unital maps $\Linf(X,\mu) \to \B(H)$;
        \item Maps $\TC(H) \to L^1(X,\mu)$ of base norm spaces.
    \end{itemize}
\end{proposition}
\begin{proof}
    Let $\varphi : \Linf(X,\mu) \to \B(H)$ be normal positive unital. The
    predual $\Linf(X,\mu)_{\#}$ of $\Linf(X,\mu)$ is isomorphic to
    $L^1(X,\mu)$, so to define a map $\TC(H) \to L^1(X,\mu)$, we can also
    define a map $\Phi_\varphi$ from $\TC(H)$ into the normal functionals on
    $\Linf(X,\mu)$. For this we take $\Phi_\varphi(T)(f) = \tr(T \varphi(f))$.
    The assignment $f \mapsto \tr(T \varphi(f))$ lies in $\Linf(X,\mu)_{\#}$,
    because $\varphi$ is normal. The map $\Phi_\varphi$ is positive since
    $\varphi$ is. To check that $\Phi_\varphi$ commutes with the maps into
    $\C$, use that the integration map $\int_X (\blank) \dd \mu : \Linf(X,\mu)
    \to \C$ corresponds to the map $\Linf(X,\mu)_{\#} \to \C$ given by $\psi
    \mapsto \psi(1)$. From this it follows that $\Phi_\varphi$ is a map of
    base norm spaces.

    Now we will show how to assign a map $ \Linf(X,\mu) \to \B(H)$ to a map
    $\Phi : \TC(H) \to L^1(X,\mu)$. First define a map $\varphi_{\Phi} :
    \Linf(X,\mu) \to \TC(H)^*$ by $\varphi_{\Phi}(f)(T) = \int_X f(x)
    \Phi(T)(x) \dd x$. This integral exists since $f$ is bounded and $\Phi(T)$
    is integrable. Since the dual of the Banach space $\TC(H)$ is isomorphic
    to the space of bounded operators on $H$, this gives a map $\Linf(X,\mu)
    \to \B(H)$, also denoted $\varphi_{\Phi}$. Positivity of $\varphi_\Phi$
    follows from positivity of $\Phi$. To show that $\varphi_\Phi$ is unital,
    note that the unit of $\Linf(X,\mu)$ is the constant function with value
    1, and the unit of $\TC(H)^*$ is the trace.  Then unitality of
    $\varphi_\Phi$ follows since $\Phi$ is a morphism of base norm spaces:
    \[ \begin{array}{c}
        \varphi_\Phi(1)(T) = \int_X \Phi(T)(x) \dd x = \tr(T).
    \end{array} \]
    The map $\varphi_\Phi$ is normal because integrals are continuous.  The
    constructions above are clearly inverses.
\end{proof}

Again we can rephrase the duality result above as an equivalence between comma
categories. First we will establish the following diagram of categories and
functors.
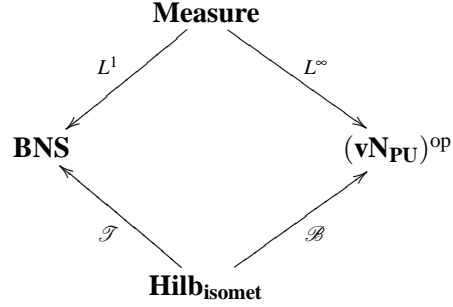
\[
\xymatrix@R+1pc { & \Measure \ar[dl]_{L^1} \ar[dr]^{\Linf} \\
\BNS & & (\CstarPU)\op  \\
& \Hilbisomet \ar[ul]^{\TC} \ar[ur]_{\B} }
\]
The functors $\TC$ and $\B$ act on morphisms via conjugation. Formally,
$\B(f)(A) = f^{\dag} A f$ and $\TC(f)(T) = f T f^{\dag}$.

The category $\Measure$ has measure spaces as objects. A morphism from
$(X,\mu)$ to $(Y,\nu)$ is a measurable map $f : X \to Y$ such that the measure
$\mu \circ f^{-1}$ on $Y$ is $\nu$-continuous, in other words, $\nu(N) = 0$
implies $\mu(f^{-1}[N]) = 0$. We have seen the action of $\Linf$ and $L^1$ on
objects before. Let $f : (X,\mu) \to (Y,\nu)$ be a morphism in $\Measure$.
Define $\Linf(f) : \Linf(Y,\nu) \to \Linf(X,\mu)$ by $\Linf(f)(\varphi) =
\varphi \circ f$. If $\varphi$ is $\nu$-essentially bounded, then $\varphi
\circ f$ is $\mu$-essentially bounded because $\mu \circ f^{-1}$ is
$\nu$-continuous. It is clear that $\Linf(f)$ is a morphism in $\CstarPU$.
To define $L^1(f)(\varphi)$ for $\varphi \in L^1(X,\mu)$, we first introduce a
new measure $\lambda$ on $Y$ via $\lambda(N) = \int_{f^{-1}[N]} \varphi \dd
\mu$. This $\lambda$ is $\nu$-continuous, so we can define $L^1(f)(\varphi)$
to be its derivative $\der{\lambda}{\nu}$. Thus $L^1(f)(\varphi)$ is the
unique function satisfying $\int_N L^1(f)(\varphi) \dd \nu = \int_{f^{-1}[N]}
\varphi \dd \mu$. Clearly $L^1(f)(\varphi)$ is integrable, and $L^1(f)$ is a
morphism in $\BNS$.

\begin{corollary}
    The categories $\comma{\TC}{L^1}$ and $\comma{\B}{\Linf}$ are equivalent.
    \label{CorPOVMCStar}
\end{corollary}
\begin{proof}
    On objects, this was established in Proposition~\ref{PropDualityCStar}. On
    morphisms, this amounts to proving naturality of the correspondence in the
    Proposition. Pick any isometry $f : H \to K$ and let $\varphi :
    \Linf(X,\mu) \to \B(K)$ be a normal positive unital map. Then we have to
    show that $\Phi_{\B(f) \circ \varphi} = \Phi_\varphi \circ \TC(f)$. This
    holds because
    \[ \Phi_{\B(f) \circ \varphi} (T)(g) = \tr(T f^{\dag} \varphi(g) f) =
    \tr(f T f^{\dag} \varphi(g)) = (\Phi_{\varphi} \circ \TC(f))(T)(g). \]
    Finally we have to prove that $\Phi_{\varphi \circ \Linf(f)}(T) = (L^1(f)
    \circ \Phi_\varphi)(T)$ for $f : (X,\mu) \to (Y,\nu)$. This is equivalent
    to showing that the integrals $\int_N \Phi_{\varphi \circ \Linf(f)}(T) \dd
    \nu$ and $\int_N (L^1(f) \circ \Phi_\varphi)(T) \dd \nu$ are equal for
    each $N$. If we identify elements of $L^1(Y)$ with normal functionals on
    $\Linf(Y)$, then integration over $N$ amounts to plugging in the
    functional $\indicator{N}$. Hence the first integral is equal to $\tr(T
    \varphi(\indicator{N} \circ f))$, and the second integral is equal to
    $\tr(T \varphi(\indicator{f^{-1}[N]}))$, thus the integrals are the same.
\end{proof}

\section{Conclusion}\label{SecConclusion}

We have established bijective correspondences between the following
representations of POVMs:
\begin{itemize}
    \item Morphisms of $\sigma$-effect algebras $\Sigma_X \to \Ef(H)$;
    \item Morphisms of $\sigma$-effect modules $\Meas(X, [0,1]) \to \Ef(H)$;
    \item Morphisms of $\Giry$-algebras $\DM(H) \to \Giry(X)$.
\end{itemize} 
In the situation where the space $X$ is compact and equipped with a finite
measure $\mu$, we obtain correspondences between the following:
\begin{itemize}
    \item POVMs $\Sigma_X \to \Ef(H)$ that are $\mu$-continuous;
    \item Normal positive unital maps $\Linf(X,\mu) \to \B(H)$;
    \item Maps $\TC(H) \to L^1(X,\mu)$ of base norm spaces.
\end{itemize}
These correspondences can be phrased as equivalences between comma categories.
The object part of these equivalences gives the bijective correspondences
above, and since we have shown that there is also an equivalence between the
morphisms of the comma categories, the above correspondences are natural.

Many POVMs occuring in physics are covariant with respect to a symmetry group
or groupoid, as discussed in \cite{Landsman98,Schroeck96}. For future
research, it would be interesting to see how our results can be extended to
the covariant setting using convolution algebras. Another possible direction
would be to study the sequential composition for POVMs in more detail, for
example by finding an axiomatization generalizing the one for effects in
\cite{GudderG02}.

\paragraph{Acknowledgements.}
This research has been financially supported by the Netherlands Organisation
for Scientific Research (NWO) under TOP-GO grant no.\ 613.001.013 (The logic
of composite quantum systems).
Thanks are due to Robert Furber and Bart Jacobs for helpful discussions and
comments.

\bibliographystyle{eptcs}

{\small
\bibliography{all}
}

\end{document}